\documentclass[journal]{IEEEtran}

\usepackage{cite,graphicx,amsmath,amssymb}
\usepackage{color,psfrag}
\usepackage{subfigure}
\usepackage{citesort}
\usepackage{psfrag}
\usepackage{multirow}
\usepackage[center]{caption2}  
\usepackage{amsthm}  
\renewcommand{\baselinestretch}{0.92}
\newtheorem{theorem}{Theorem}

\newtheorem{lemma}{Lemma}

\newtheorem{corollary}{Corollary}

\newtheorem{proposition}{Proposition}

\newtheorem{remark}{Remark}

\newcommand{\tE}{{\tt E}}

\begin{document}
%
\title{Outage Performance for Cooperative NOMA Transmission with an AF Relay}
%
%
 \author{Xuesong~Liang,
         Yongpeng~Wu,~\IEEEmembership{Senior Member,~IEEE,}
        Derrick~Wing~Kwan~Ng,~\IEEEmembership{Member,~IEEE,}
        Yiping~Zuo,
        Shi~Jin,~\IEEEmembership{Member,~IEEE}
        and~Hongbo~Zhu
\thanks{Manuscript received December 28, 2016; revised February 19, 2017; accepted March 5, 2017. The work of X.~Liang was supported by the Jiangsu Provincial Colleges and Universities Project (Grant No. 16KJB510026) and was sponsored by NUPTSF (Grant No. NY213062, NY214035). The work of Y.~Wu was supported by the TUM University Foundation Fellowship. The work of D. W. K. Ng was supported under Australian Research Council¡¯s Discovery Early Career Researcher  Award funding scheme (Project NO. DE170100137). The work of S.~Jin was supported in part by the National Science Foundation (NSFC) for Distinguished Young Scholars  of China with Grant 61625106.
}
\thanks{Xuesong Liang, Yiping~Zuo and Hongbo~Zhu are with the School of Telecommunication and Information Engineering, Nanjing University of Posts and Telecommunications, Nanjing, China (e-mail: liangxs@njupt.edu.cn, zuoyiping93@gmail.com and zhuhb@njupt.edu.cn).

Yongpeng~Wu is with Institute for Communications Engineering,  Technical University of Munich,
, Germany (email: yongpeng.wu2016@gmail.com).

Derrick Wing Kwan Ng is with the School of Electrical Engineering and Telecommunications,
the University of New South Wales, Australia (email: w.k.ng@unsw.edu.au).

Shi~Jin is with the National Mobile Communications Research
Laboratory, Southeast University, Nanjing, China (e-mail: jinshi@seu.edu.cn).
}
}

%
%

\markboth{IEEE Communications Letters,~Vol.~?, No.~?, ?~2017}
{Shell \MakeLowercase{\textit{et al.}}: Bare Demo of IEEEtran.cls for IEEE Communications Society Journals}

%



\maketitle

\begin{abstract}
This letter studies the outage performance of cooperative non-orthogonal multiple access  (NOMA)
network with the help of an amplify-and-forward  relay.
An accurate closed-form approximation for the exact outage probability is derived.
Based on this, the asymptotic outage probability  is investigated,
which shows that cooperative NOMA achieves the same diversity order and the superior coding gain compared to cooperative orthogonal multiple access.
It is also revealed that when the transmit power of  relay is smaller than that of the base station,  the outage performance improves as the distance between the relay and indirect link user decreases.
\end{abstract}


%
\IEEEpeerreviewmaketitle

\section{Introduction}
%
%
%
%
\IEEEPARstart{I}n recent years, non-orthogonal multiple access (NOMA), which is a potential candidate for the multiple access schemes in the fifth-generation
(5G), has received significant attention \cite{QCLi_5G,YSaito_Nonorthogonal}.
Different from the conventional orthogonal multiple access (OMA) schemes, NOMA can serve multiple users with utilizing the same time and frequency resource to obtain a significant spectral efficiency gain \cite{ABenjebbour_Concept}.
As one of the key radio access technologies to achieve higher spectral efficiency with low cost in 5G, NOMA is extended to cooperative transmission to enhance the transmission reliability for the users with poor channel conditions \cite{ZDing_Cooperative}.
In particular, a cooperative NOMA scheme was proposed in \cite{ZDing_Cooperative} by selecting the users with better channel conditions as relays for assisting the transmission of other users.
In addition, a cooperative NOMA scheme with dedicated relays was investigated in \cite{ZDing_Relay}
where a two-stage relay selection strategy for NOMA was proposed.
The concept of cooperative NOMA with  dedicated relays was also extended to systems with multiple users equipped with multiple antennas \cite{JMen_Nonorthogonal}.
Specifically, relay selection based on the maximal instantaneous signal-to-noise-ratio (SNR) was proposed and the corresponding system outage performance was analyzed.
The aforementioned works on cooperative NOMA schemes with dedicated relays usually assume that there are no direct links between the base station (BS) and the users, and all the users must exchange information with the BS through the relays. However, for typical scenarios of small cells in 5G networks \cite{YNiu_Exploiting},
some users can  directly communicate with the BS while some cannot.
In fact the research on cooperative NOMA schemes taking into both direct and indirect links users is also addressed. In particular, a recent work addressed in \cite{JBKim_Nonorthogonal} shows that the spectral efficiency is remarkably improved when coordinated direct and decode-and-forward relaying was employed in NOMA scheme,
and  a  cooperative NOMA system with  a dedicated full-duplex relay was proposed in \cite{CZong_Nonorthogonal_FD}, which provided the exact analytical expressions for achievable outage probability  and ergodic sum capacity of the system.
However, to the best of the authors' knowledge, the results for amplify-and-forward (AF) relaying NOMA systems has not yet been reported, which motivates the study of this letter.

In this letter, we investigate a downlink cooperative NOMA scheme including direct and indirect link users, where a dedicated relay node with AF protocol is adopted.
Our contributions include two parts: 1) We compare the overall outage probability of the cooperative NOMA with the conventional cooperative OMA, which indicates that the cooperative NOMA outperforms cooperative OMA significantly. Moreover, we derive a closed-form approximation of the  outage probability, which is shown to closely match with the exact results;  2) We analyze the outage performance in the high SNR regime, where both the diversity order and coding gain of the system are derived.
We show that the diversity order of the system is one and the coding gain is affected by the location of the relay.
Besides, if the transmit power of the relay is much smaller than the transmit power of the BS, the outage performance is improved when the relay is nearer to indirect link user.
Monte Carlo simulations validate our  analytical results.

\emph{Notations}---
In this letter,  we denote the probability and the expectation  value of a random event $A$ by ${\tt P}\{A\}$ and ${\tt E}\{A\}$, respectively.
$| \cdot |$ denotes the Euclidean norm of a scalar.
%
\section{System Model}\label{System_Model}
Consider a fundamental model of a downlink cooperative NOMA system 
including one BS, two users (UE1 and UE2), and one  relay node (R), in which UE1 directly communicates with  the BS.
 Besides, UE2 needs the help from R because there is no direct path between the  BS and UE2 due to the long distance or significant blockage between them. Each node is equipped with a single-antenna and the relay operates in half-duplex mode using an AF protocol. The scheme of cooperative NOMA consists of two consecutive equal length time slots, as described in the follows.

During the first time slot, the BS broadcasts the superimposed signal, $x_1  + x_2$,
to R and UE1,
 where $x_1$ and $x_2$ are the corresponding signals for UE1 and UE2, respectively,  with ${\tt E}\{\left| {x_1 } \right|^2\}=P_1$ and ${\tt E}\{\left| {x_2 } \right|^2\}=P_2$. According to the NOMA protocol described in \cite{ABenjebbour_Concept}, we set $P_1 < P_2$ and denote the total transmit power for the BS as $ P_T = P_1  + P_2 $.\\
 Therefore the received signals $y'_1 $ at UE1, and $y_r$ at R are given by
 \begin{small}
\begin{align}
&y'_1  = h_1 x + n_1,\label{Ph1:y1}\\
&y_r  = h_r x + n_r,\label{Ph1:yr}
\end{align}\end{small}respectively,
where $n_1 $ and $n_r $ denote the complex additive white Gaussian noises (AWGN), both with zero mean and variance $N_0$ at UE1 and R, respectively.
 Meanwhile, we assume that the channels between the BS and  UE1, $h_1 $, and that between the BS and R, $h_r $, are independent
Rayleigh fading  with ${\tt E}\{\left| {h_1 } \right|^2\}=\sigma^2_1$ and ${\tt E}\{\left| {h_r  } \right|^2\}=\sigma^2_r$.

During the second time slot, the BS remains silent and R broadcasts the signal $x_r$ to UE1 and UE2 after multiplying  the previous received signal, $y_r$  with an amplifying gain $\rho   = \sqrt {\frac{{P_R }}{{P_T \left| {h_r } \right|^2  + N_0 }}}$  \cite{CSPatel_Statistical}, where $P_R  = \tE\left\{ {\left| {x_r } \right|^2 } \right\}$
denotes the transmit power of the relay. Therefore,
the received signals $y''_1$ at UE1, and $y_2$ at UE2 are given by
\begin{small}\begin{align}
&y''_1  
= \tilde h_1 x_r + \tilde n_1,\label{Ph2:y1}\\
&y_2  = h_{r,2} x_r  + n_2,\label{Ph2:y2}
\end{align}\end{small}with
$x_r  = \rho y_r  $, $\tilde h_1  = \rho h_{r,1} h_r$ and $\tilde n_1  = \rho h_{r,1} n_r  + n'_1$.
We  denote the AWGNs both with zero mean and  variance $N_0$ at UE1 (in relaying phase) and UE2 by $n'_1  $ and $n_2  $, respectively.
We also assume that the channels between R and UE1, $h_{r,1} $, and that between R and UE2, $ h_{r,2}  $, are independent
Rayleigh fading  with ${\tt E}\{\left| {h_{r,1} } \right|^2\}= \sigma _{r,1}^2 $ and ${\tt E}\{\left| {h_{r,2} } \right|^2\}= \sigma _{r,2}^2$.\\
{{Then, by using the maximum ratio combining criterion, UE1 combines the received signals, $y'_1$ and $y''_1$,  with the conjugate of $h_1$ and $\tilde h_1$, respectively, 
which yields}}
\begin{small}\begin{equation}\label{MRC:y1}
y_1  = h_1^ *  y'_1  + \tilde h_1^ *  y''_1.
\end{equation}\end{small}According to
the NOMA scheme in \cite{ABenjebbour_Concept}, {{UE1 firstly decodes the data of $x_2$. After $x_2$ is decoded successfully, UE1 removes the signal of $x_2$ and then decodes the data of $x_1$ based on successive interference cancellation (SIC).}}
 Therefore the signal-to-interference-plus-noise-ratios (SINRs) for decoding $x_2$ and $x_1$
by UE1 are respectively given by
\begin{small}
\begin{align}
\gamma _{12}  &= \frac{{\left( {\left| {h_1 } \right|^2  + \left| {\tilde h_1 } \right|^2 } \right)^2 P_2 }}{{\left( {\left| {h_1 } \right|^2  + \left| {\tilde h_1 } \right|^2 } \right)^2 P_1  + \left[ {\left| {\tilde h_1 } \right|^2 \left( {\rho ^2 \left| {h_{r,1} } \right|^2  + 1} \right) + \left| {h_1 } \right|^2 } \right]N_0 }},\label{SINR:y1r}\\
\gamma _1 & = \frac{{\left( {\left| {h_1 } \right|^2  + \left| {\tilde h_1 } \right|^2 } \right)^2 P_1 }}{{\left| {h_1 } \right|^2 N_0  + \left| {\tilde h_1 } \right|^2 \left( {\rho ^2 \left| {h_{r,1} } \right|^2  + 1} \right)N_0 }},\label{SINR:y1}
\end{align}
\end{small}and
the SINR for decoding $x_2$ by UE2 is given by
\begin{small}\begin{equation}\label{SINR:y2}
\gamma _2  = \frac{{\left| {h_{r,2} } \right|^2 \left| {h_r } \right|^2 P_2 }}{{\left| {h_{r,2} } \right|^2 \left| {h_r } \right|^2 P_1  + \left( {\left| {h_{r,2} } \right|^2  + \rho ^{ - 2} } \right)N_0 }}.
\end{equation}\end{small}

\section{Outage Performance Analysis}\label{Performance_Analysis}
In this section, we investigate the outage probability, which is an important metric of the considered cooperative NOMA system. However, the analytical expression of the exact outage probability for this system is mathematically intractable. Alternatively, an approximation with a closed-form expression for the outage probability is derived, and based on which the asymptotic characteristics of outage performance in the high SNR regime are also discussed.

\subsection{Outage Probability}\label{Outage_Probability}
To begin with, we characterize the outage probability achieved by this two-phase cooperative NOMA system. Denoting the data rate requirements for UE1 and UE2 as $R_1$ and $R_2$, respectively, we note that the overall outage probability of system is defined as
\begin{small}
\begin{equation}\label{Def:P_out}
{\tt P}_{\rm out}  \buildrel \Delta \over = {\tt P}\left\{ {\gamma _{12}  < f\left( {R_2 } \right){\mbox{ or }}\gamma _1  < f\left( {R_1 } \right){\mbox{ or }}\gamma _2  < f\left( {R_2 } \right)} \right\}
\end{equation}
\end{small}
{{where $f\left( R \right) $ denotes the SINR threshold relative to the practical data rate requirement with $f\left( R \right) = 2^{2R}  - 1$.
In the definition of (\ref{Def:P_out}), we note that $\gamma _1  $, $\gamma _{12} $ and $\gamma _2  $ denote  the SINRs of  $x_1$ and $x_2$ at UE1, and SINR of $x_2$ at UE2, respectively, while $ f\left( {R_1 } \right)$ and $ f\left( {R_2 } \right)$ represent the SINR thresholds for successfully decoding $x_1$ and $x_2$,  respectively.}}
Unfortunately, mathematical analysis of the outage probability in (\ref{Def:P_out}) becomes intractable since the considered random events are correlated.
Hence, an alternative with a closed-form approximation to (\ref{Def:P_out}) is needed.

\begin{proposition}\label{Thm:Approx_PoutL}
An approximation for the outage probability of the system is given by (\ref{Eqn:PoutA}) (on top of next page)
\begin{figure*}
\begin{small}\begin{align}\label{Eqn:PoutA}
 {\tt P}_{\rm out}^{\rm A}   = \left\{ {\begin{aligned}
 & {{1 - \mu \exp \left( { - \delta } \right)K_1 \left( \mu  \right)\left( {1 + \frac{{\bar \gamma _1 }}{{\sigma _1^2 }}} \right)\exp \left( { - \frac{{\bar \gamma _1 }}{{\sigma _1^2 }}} \right)},}
  {\mbox{ when } \sigma _1^2  = \lambda \sigma _{r,1}^2 {{{\mbox{ and }}P_2 } \mathord{\left/
 {\vphantom {{{}P_2 } {P_1 }}} \right.
 \kern-\nulldelimiterspace} {P_1 }} > f\left( {R_2 } \right)}  \\
&1 - \mu \exp \left( { - \delta } \right)K_1 \left( \mu  \right)\left[ {\eta \exp \left( { - \frac{{\bar \gamma _1 }}{{\sigma _1^2 }}} \right) + \left( {1 - \eta } \right)\exp \left( { - \frac{{\bar \gamma _1 }}{{\lambda \sigma _{r,1}^2 }}} \right)} \right],
   \mbox{ when } \sigma _1^2  \ne \lambda \sigma _{r,1}^2 {{{\mbox{ and }}P_2 } \mathord{\left/
 {\vphantom {{{}P_2 } {P_1 }}} \right.
 \kern-\nulldelimiterspace} {P_1 }} > f\left( {R_2 } \right) \\
& 1,{\mbox{ else}}{{.}} \\
\end{aligned}} \right.
\end{align}\end{small}
\hrulefill
\end{figure*}
with $\eta  = \frac{{\sigma _1^2 }}{{\sigma _1^2  - \lambda \sigma _{r,1}^2 }}$,
$\mu  =\frac{{2\theta _r }}{{\sqrt {\lambda \sigma _r^2 \sigma _{r,2}^2 } }}$, $\delta  = \frac{{\theta _r }}{{\sigma _r^2 }} + \frac{{\theta _r }}{{\lambda \sigma _{r,2}^2 }}$, and $\lambda  = \frac{{P_R }}{{P_T }}$.
\end{proposition}
\begin{proof}\label{Proof_Thm_PoutA}
Note that $ \gamma _{12}$ is superior to $ \gamma _2 $ when $
\left| {h_1 } \right|^2  + \left| {\tilde h_1 } \right|^2  > \rho ^2 \left| {h_{r,2} } \right|^2 \left| {h_r } \right|^2
$. This condition is always satisfied in practical systems, hence we neglect $\gamma _{12}$ 
and further obtain a lower bound of ${\tt P}_{\rm out}$ as
\begin{small}\begin{align}\label{Eqn0:PoutL}
{\tt P}_{\rm out} \ge {\tt P}_{\rm out} ^{\rm L} = {\tt P}\left\{ {\gamma _1^{\rm L}  < f\left( {R_1 } \right){\rm{~ or~ }}\gamma _2  < f\left( {R_2 } \right)} \right\},
\end{align}\end{small}%
with $ \gamma _1^{\rm L}    =\frac
{{\left( {\left| {h_1 } \right|^2  + \left| {\tilde h_1 } \right|^2 } \right)P_1 }}{{N_0 }}$.
The expression of (\ref{Eqn0:PoutL}) is much simpler when  it is compared to (\ref{Def:P_out}),
but further simplification of (\ref{Eqn0:PoutL}) is still needed.
Hence, we use an approximation of $\rho^2$ for the medium-high SNR and obtain a closed-form expression of  $ {\tt P}_{\rm out}^{\rm A} $ (Although $ {\tt P}_{\rm out}^{\rm A} $ is an approximation for ${\tt P}_{\rm out} ^{\rm L}$ and thus an approximation for ${\tt P}_{\rm out} $, it is sufficiently accurate for the exact results as shown as in Section~\ref{Numerical_Results}) is given by (\ref{Eqn:PoutA})
 (on the top of next page) and the proof is shown in Appendix \ref{Appendix:PoutA}.
\end{proof}
\begin{remark}\label{Remark1:PoutA}
Proposition~\ref{Thm:Approx_PoutL} shows that the outage  of system definitely occurs (${\tt P}_{\rm out}^{\rm A} =1$) when $ {{P_2 } \mathord{\left/
 {\vphantom {{P_2 } {P_1 }}} \right.
 \kern-\nulldelimiterspace} {P_1 }} \le f\left( {R_2 } \right)$.
The reason for the outage is that when $ {{P_2 } \mathord{\left/
 {\vphantom {{P_2 } {P_1 }}} \right.
 \kern-\nulldelimiterspace} {P_1 }} \le f\left( {R_2 } \right)$,
 $x_2$ is dominated by  the interference from $x_1$ and hence cannot be decoded successfully by UE2 or UE1.
Therefore Proposition~\ref{Thm:Approx_PoutL} reveals that 
 the minimum transmit power that should be assigned for $x_2$ by BS is $
{{P_T  \cdot f\left( {R_2 } \right)} \mathord{\left/
 {\vphantom {{P_T  \cdot f\left( {R_2 } \right)} {\left( {1 + f\left( {R_2 } \right)} \right)}}} \right.
 \kern-\nulldelimiterspace} {\left( {1 + f\left( {R_2 } \right)} \right)}}$.
\end{remark}
\begin{remark}\label{Remark2:PoutA}
Note that $ {\tt P}_{\rm out}^{\rm A} $ is represented by ${\tt P}_{\rm out}^{\rm A}= 1-{\tt P}_{\rm A}\cdot {\tt P}_{\rm B}$, where
\begin{small}
\begin{align}
{\tt P}_{\rm A}
 &= \left\{ {\begin{aligned}
   &{{\left( {1 + \frac{{\bar \gamma _1 }}{{\sigma _1^2 }}} \right)\exp \left( { - \frac{{\bar \gamma _1 }}{{\sigma _1^2 }}} \right)},\sigma _1^2  = \lambda \sigma _{r,1}^2 }  \\
   &{\eta \exp \left( { - \frac{{\bar \gamma _1 }}{{\sigma _1^2 }}} \right) + \left( {1 - \eta } \right)\exp \left( {  \frac{{\bar \gamma _1 }}{{\lambda \sigma _{r,1}^2 }}} \right),\sigma _1^2  \ne \lambda \sigma _{r,1}^2 }  \\
\end{aligned}} \right.,\label{Eqn:P_A}\\
{\tt P}_{\rm B}& = \mu \exp \left( { - \delta } \right)K_1 \left( \mu  \right),\label{Eqn:P_B}
\end{align}
\end{small}in which
$K_1 \left( \mu  \right)$ denotes the first order modified Bessel function of the second kind.
In physics, ${\tt P}_{\rm A}$ and ${\tt P}_{\rm B}$ represent the probabilities of $x_1$ being decoded by UE1, and  $x_2$ being decoded by UE2, respectively.
Therefore, we may use $1-{\tt P}_{\rm A}$ and $1-{\tt P}_{\rm B}$ to estimate the outage probabilities for UE1 and UE2, respectively.
\end{remark}
\subsection{Asymptotic Analysis}\label{Asymptotic_Behavior}
 In this subsection, we focus on the high SNR regime and discuss the asymptotic characteristics of the outage probability based on (\ref{Eqn:PoutA}). To analyze ${\tt P}_{\rm out}^{\rm A}$ in the  high SNR regime, we fix
${{P_R } \mathord{\left/
 {\vphantom {{P_R } {P_T }}} \right.
 \kern-\nulldelimiterspace} {P_T }}$ and
 ${{P_1 } \mathord{\left/
 {\vphantom {{P_1 } {P_T }}} \right.
 \kern-\nulldelimiterspace} {P_T }}$ when
 ${{P_T } \mathord{\left/
 {\vphantom {{P_T } {N_0 }}} \right.
 \kern-\nulldelimiterspace} {N_0 }} \to \infty$
  (with denoting  $\bar \gamma _0  = {{P_T } \mathord{\left/
 {\vphantom {{P_T } {N_0 }}} \right.
 \kern-\nulldelimiterspace} {N_0 }}$).
 Then, we obtain the following proposition.

\begin{proposition}\label{Thm:High_SNR}
When $\bar \gamma _0 \to \infty$ and  ${{P_2 } \mathord{\left/
 {\vphantom {{{}P_2 } {P_1 }}} \right.
 \kern-\nulldelimiterspace} {P_1 }} > f\left( {R_2 } \right)$, the asymptotic expression of the outage probability  is given by
\begin{small}\begin{align}\label{Eqn:High_SNR}
\mathop {\lim }\limits_{\bar \gamma _0  \to \infty } {\tt P}_{\rm out}^{\rm A} \approx
\frac{{\delta _r }}{{\bar \gamma _0 }},
\end{align}\end{small}where
$\delta _r = \frac{{f\left( {R_2 } \right)}}{{1 - \lambda _1 \left[ {1 + f\left( {R_2 } \right)} \right]}}\left( {\frac{1}{{\sigma _r^2 }} + \frac{1}{{\lambda \sigma _{r,2}^2 }}} \right)$ and $ {{\lambda _1  = P_1 } \mathord{\left/
 {\vphantom {{\lambda _1  = P_1 } {P_T }}} \right.
 \kern-\nulldelimiterspace} {P_T }}$.
\end{proposition}

\begin{proof}
{{The asymptotic characteristic of ${K_1  \left( x \right)}$ in \cite[(9.7.2)]{Handbook_functions} shows that $K_1 \left( x \right) \to {1}/{x}$ when $x$ is sufficiently small.
Therefore we recast ${\tt P}_{\rm out}^{\rm A} $ in (\ref{Aeqn:PoutA}) by substituting ${K_1  \left( x \right)}$ with ${1}/{x}$.
Then, by using Taylor expansion for ${\exp \left(  {x}\right)}$ at $0$ and omitting the high order terms of $x$ when $x \to 0$, we finally obtain the approximation of ${\tt P}_{\rm out}^{\rm A} $ as (\ref{Eqn:High_SNR})
(the tedious details of the derivation are  omitted due to the limitation of space.).}}
\end{proof}

\begin{remark}\label{Remark1:Asymp}
With the definition of diversity order, $d =   - \mathop {\lim }\limits_{\bar \gamma _0  \to \infty } \frac{{\log \left( {\tt P}_{\rm out}^{\rm A}\right)}}{{\log \bar \gamma _0 }}$,
 we show that the diversity order of this cooperative NOMA system is one.
\end{remark}

\begin{remark}\label{Remark2:Asymp}
{{Based on Proposition~\ref{Thm:High_SNR}, we note that the value of $\delta _r $  is determined by $\sigma _r^2 $ and $ \lambda \sigma _{r,2}^2$
with any given $\lambda_1$ and $R_2$.
Specifically, if we  fix $\sigma _r^2 \cdot \sigma _{r,2}^2$, then the outage performance is better when $\sigma _r^2 < \sigma _{r,2}^2$ than that when $\sigma _r^2 \ge \sigma _{r,2}^2$ because $\lambda<1$ usually holds  in practical systems. The physical meaning of that is if the sum distance  of the relay to UE2 and the relay to the BS is fixed,  the outage performance when the relay is close to  UE2 is superior to that when the relay is
close to the BS. This is because the power of relay can be  full utilized to reduce the outage probability of system when the relay is close the user.}}
\end{remark}

\section{Numerical Results}\label{Numerical_Results}

\begin{figure}[htbp]
\centering
\includegraphics[scale=0.65]{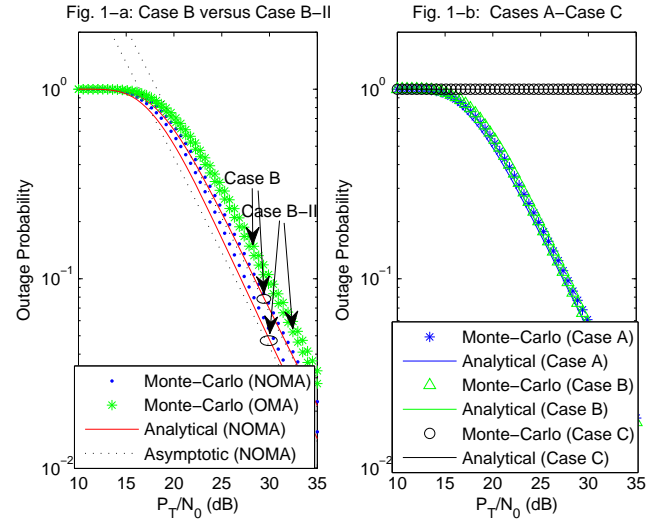}
\caption{{{The outage probability versus SNR.}}}
\label{fig:result_2}
\end{figure}

\begin{figure*}[ht]
\begin{small}
\begin{align}\label{Aeqn0:Part1}
{\tt P}\left\{ {\left| {h_1 } \right|^2  + \lambda \left| {h_{r,1} } \right|^2 < \bar \gamma _1 } \right\} = \left\{ {\begin{aligned}
   &{1 -{\left( {1 + \frac{{\bar \gamma _1 }}{{\sigma _1^2 }}} \right)\exp \left( { - \frac{{\bar \gamma _1 }}{{\sigma _1^2 }}} \right)},\sigma _1^2  = \lambda \sigma _{r,1}^2 }  \\
   &{1 - \eta \exp \left( { - \frac{{\bar \gamma _1 }}{{\sigma _1^2 }}} \right) - \left( {1 - \eta } \right)\exp \left( { - \frac{{\bar \gamma _1 }}{{\lambda \sigma _{r,1}^2 }}} \right),\sigma _1^2  \ne \lambda \sigma _{r,1}^2 }  \\
\end{aligned}} \right.
\end{align}
\end{small}
\hrulefill
\end{figure*}
\begin{figure*}[ht]
\begin{small}
\begin{align}\label{Aeqn0:Part2}
& {\tt P}\left\{ {\lambda \left| {h_{r,2} } \right|^2 \left( {1 - \frac{{\theta _r }}{{\left| {h_r } \right|^2 }}} \right) \le \theta _r } \right\}
=1 - \exp \left( { - \frac{{\theta _r }}{{\sigma _r^2 }} - \frac{{\theta _r }}{{\lambda \sigma _{r,2}^2 }}} \right)\mu K_1 \left( \mu  \right)
\end{align}
\end{small}
\hrulefill
\end{figure*}
\begin{figure*}
\begin{small}
\begin{align}\label{Aeqn:PoutA}
&{\tt P}\left\{ {\left| {h_1 } \right|^2  + \lambda \left| {h_{r,1} } \right|^2  < \bar \gamma _1  {\mbox{ or }}{\lambda \left| {h_{r,2} } \right|^2 \left( {1 - \frac{{\theta _r }}{{\left| {h_r } \right|^2 }}} \right) \le \theta _r }} \right\}
 &= \left\{ {\begin{aligned}
   &{1 - \mu \exp \left( { - \delta } \right)K_1 \left( \mu  \right){\left( {1 + \frac{{\bar \gamma _1 }}{{\sigma _1^2 }}} \right)\exp \left( { - \frac{{\bar \gamma _1 }}{{\sigma _1^2 }}} \right)},\sigma _1^2  = \lambda \sigma _{r,1}^2 }  \\
 & { 1 - \mu \exp \left( { - \delta } \right)K_1 \left( \mu  \right)\left[ {\eta \exp \left( { - \frac{{\bar \gamma _1 }}{{\sigma _1^2 }}} \right) + \left( {1 - \eta } \right)\exp \left( { - \frac{{\bar \gamma _1 }}{{\lambda \sigma _{r,1}^2 }}} \right)} \right],\mbox{else}}
\end{aligned}} \right.
\end{align}
\end{small}
\hrulefill
\end{figure*}

In this section, the outage probability of this cooperative NOMA system is evaluated based on Monte-Carlo simulations averaging over $10^5$ independent channel realizations {{ whilst the noise power are set as $N_0=1$ and $P_T$ varies from $10$~dB to $35$~dB.
We define the three cases of (\ref{Eqn:PoutA})  as following: Case A when  $\sigma _1^2  = \lambda \sigma _{r,1}^2$ and $P_1/P_2> f\left( {R_2 } \right) $;
Case B when  $\sigma _1^2  \ne \lambda \sigma _{r,1}^2$ and $P_1/P_2> f\left( {R_2 } \right) $; Case C when  $P_1/P_2\le f\left( {R_2 } \right) $.
Without loss of generality, we compare outage performance for Case B of cooperative NOMA scheme with that for cooperative OMA scheme.
Since cooperative NOMA scheme includes two equal-length time slots whilst cooperative OMA scheme includes three equal-length time slots
 during the whole transmission, we set the data requirements for cooperative OMA as $1.5$-fold of that for cooperative NOMA for a fair comparison.  Except for that, the  parameters for cooperative OMA are identical with that for cooperative NOMA.
  The data rate requirements for cooperative NOMA are set as $R_1 =1$ and $ R_2=0.7$,
and the distances of BS-UE1 and R-UE1 are set as $d_1=30$~m and  $d_{r,1}=30$~m, respectively. Besides, the variances of the channels are calculated by $\sigma _i^2  = \left( {{{d_i } \mathord{\left/ {\vphantom {{d_i } {d_0 }}} \right. \kern-\nulldelimiterspace} {d_0 }}} \right)^{ - \alpha }$ with $d_0 =20$~m and $\alpha =2$.
For comparing the outage performance for different relay locations, we consider the following two scenarios of Case B:
The distances of BS-R and R-UE2 are set as $d_r=30$~m and $d_{r,2}=45$~m,  respectively for Case B-I, and they are set as $d_r=45$~m, $d_{r,2}=30$~m, respectively for Case B-II.
 Moreover, the factors of transmit power ratios are set as $\lambda = 0.3$ and $ \lambda_1 = 0.2$ for NOMA,  while $ \lambda_1 $ for OMA is optimized through the brute-force search to obtain the minimal outage probability.

In Fig.~\ref{fig:result_2}-a, the Monte-Carlo results of outage probabilities  for cooperative NOMA are compared with the approximate results, and also with the simulation results  for cooperative OMA.
We observe from Fig.~\ref{fig:result_2}-a that the proposed approximation shows an excellent agreement with the exact outage
probability for cooperative NOMA, especially for the high SNR regime.
Our results also illustrate that the outage performance for cooperative NOMA outperforms that for cooperative OMA for all SNR regimes.
 It is shown that both cooperative NOMA and cooperative OMA achieve the same diversity order for high SNR regime, which is predicted to be $1$ in Proposition \ref{Thm:High_SNR}. Meanwhile, the coding gain for cooperative NOMA is superior to that for cooperative OMA, which shows the advantage of outage performance for cooperative NOMA.
Besides, the impact of relay location on coding gain is also examined by comparing the simulation results of Cases B-I and B-II.
 Fig.~\ref{fig:result_2}-a shows that the coding gain is improved when R is closer to  UE2. 
In Fig.~\ref{fig:result_2}-b, the outage probability for Cases A, B, and C of cooperative NOMA are shown.
The parameters  for Cases A and C are identical to that of Case B-I with an exception of $d_{r,1}\approx16.43$~m for Case A, and $\lambda_1 =0.4$ for Case C.
  The analytical results show good match  with numerical results for each case of (\ref{Eqn:PoutA}) in Fig.~\ref{fig:result_2}-b.}}

\section{Conclusion}
This letter studied the outage performance for a downlink cooperative NOMA scenario with the help of an AF relay. The approximation for outage probability of the system was derived in a closed-form expression and {{the accuracy of the approximation is verified by various numerical simulations.  Furthermore, the asymptotic behaviors for the considered system were investigated  for the high SNR regime, which indicates that the cooperative NOMA is obviously superior to cooperative OMA in coding gain without losing diversity order. Simulation results also examine the impact of  relay location variations on coding gain of the system.}}

\appendices
\section{Proof of Proposition \ref{Thm:Approx_PoutL}}\label{Appendix:PoutA}
With the approximation of $ \rho ^2   \approx \frac{{P_R }}{{P_T \left| {h_r } \right|^2 }}$ for medium-high SNR,  
(\ref{Eqn0:PoutL}) is represented as
\begin{small}
\begin{align}\label{Aeqn0:PoutA}
 &{\tt P}_{\rm out} ^{\rm L} \approx{\tt P}\left\{ {\left| {h_1 } \right|^2  + \lambda \left| {h_{r,1} } \right|^2  < \bar \gamma _1  {\mbox{ or }}{\lambda \left| {h_{r,2} } \right|^2 \left( {1 - \frac{{\theta _r }}{{\left| {h_r } \right|^2 }}} \right) \le \theta _r }} \right\} \nonumber\\
&= {\tt P}\left\{ {\left| {h_1 } \right|^2  + \lambda \left| {h_{r,1} } \right|^2  < \bar \gamma _1 } \right\} 
  + {\tt P}\left\{ {\lambda \left| {h_{r,2} } \right|^2 \left( {1 - \frac{{\theta _r }}{{\left| {h_r } \right|^2 }}} \right) \le \theta _r } \right\} \nonumber\\
 &-  {\tt P}\left\{ {\left| {h_1 } \right|^2  + \lambda \left| {h_{r,1} } \right|^2  < \bar \gamma _1 } \right\} {\tt P}\left\{ {\lambda \left| {h_{r,2} } \right|^2 \left( {1 - \frac{{\theta _r }}{{\left| {h_r } \right|^2 }}} \right) \le \theta _r } \right\}
 \end{align}
 \end{small}and the main terms in (\ref{Aeqn0:PoutA}) are calculated as follows.

The first term in (\ref{Aeqn0:PoutA}) is given by (\ref{Aeqn0:Part1}) (on the top of the page)
with $\eta  = \frac{{\sigma _1^2 }}{{\sigma _1^2  - \lambda \sigma _{r,1}^2 }} $.
In the sequel, by utilizing \cite[(3.324-1)]{Table_Integral}, the second term in (\ref{Aeqn0:PoutA}) is given by (\ref{Aeqn0:Part2}) (on the top of the page)
with $ \lambda  = {{P_R } \mathord{\left/
 {\vphantom {{P_R } {P_T }}} \right.
 \kern-\nulldelimiterspace} {P_T }}$ and $ \mu  = \frac{{2\theta _r }}{{\sqrt {\lambda \sigma _r^2 \sigma _{r,2}^2 } }}
 $. Finally, by substituting (\ref{Aeqn0:Part1}) and (\ref{Aeqn0:Part2}) into (\ref{Aeqn0:PoutA}), we obtain (\ref{Aeqn:PoutA}) (on the top of the page)
   with $\delta  = \frac{{\theta _r }}{{\sigma _r^2 }} + \frac{{\theta _r }}{{\lambda \sigma _{r,2}^2 }}$, which completes the proof.

%

%
%

\ifCLASSOPTIONcaptionsoff
  \newpage
\fi


\begin{thebibliography}{10}
\providecommand{\url}[1]{#1}
\csname url@samestyle\endcsname
\providecommand{\newblock}{\relax}
\providecommand{\bibinfo}[2]{#2}
\providecommand{\BIBentrySTDinterwordspacing}{\spaceskip=0pt\relax}
\providecommand{\BIBentryALTinterwordstretchfactor}{4}
\providecommand{\BIBentryALTinterwordspacing}{\spaceskip=\fontdimen2\font plus
\BIBentryALTinterwordstretchfactor\fontdimen3\font minus
  \fontdimen4\font\relax}
\providecommand{\BIBforeignlanguage}[2]{{%
\expandafter\ifx\csname l@#1\endcsname\relax
\typeout{** WARNING: IEEEtran.bst: No hyphenation pattern has been}%
\typeout{** loaded for the language `#1'. Using the pattern for}%
\typeout{** the default language instead.}%
\else
\language=\csname l@#1\endcsname
\fi
#2}}
\providecommand{\BIBdecl}{\relax}
\BIBdecl
\bibitem{QCLi_5G}
Q.~C. Li, H.~Niu, A.~T. Papathanassiou, and G.~Wu, ``{5G} network capacity: Key
  elements and technologies,'' \emph{{IEEE} Veh. Technol. Mag.}, vol.~9, no.~1,
  pp. 71--78, Mar. 2014.

\bibitem{YSaito_Nonorthogonal}
Y.~Saito, Y.~Kishiyama, A.~Benjebbour, T.~Nakamura, A.~Li, and K.~Higuchi,
  ``Non-orthogonal multiple access {(NOMA)} for cellular future radio access,''
  in \emph{Proc. {IEEE} Vehicular Technology Conference {(VTC)}}, Dresden,
  Germany, Jun. 2013, pp. 1--5.

\bibitem{ABenjebbour_Concept}
A.~Benjebbour, Y.~Saito, Y.~Kishiyama, A.~Li, A.~Harada, and T.~Nakamura,
  ``{Concept and practical considerations of non-orthogonal multiple access
  (NOMA) for future radio access},'' in \emph{Proc. Intelligent Signal
  Processing and Communications Systems (ISPACS)}, Okinawa, Japan, Nov. 2013,
  pp. 770--774.

\bibitem{ZDing_Cooperative}
Z.~Ding, M.~Peng, and H.~V. Poor, ``Cooperative non-orthogonal multiple access
  in {5G} systems,'' \emph{{IEEE} Commun. Lett.}, vol.~19, no.~8, pp.
  1462--1465, Aug. 2015.

\bibitem{ZDing_Relay}
Z.~Ding, H.~Dai, and H.~V. Poor, ``Relay selection for cooperative {NOMA},''
  \emph{{IEEE} Commun. Lett.}, vol.~5, no.~4, pp. 416--419, Aug. 2016.

\bibitem{JMen_Nonorthogonal}
J.~Men and J.~Ge, ``Non-orthogonal multiple access for multiple-antenna
  relaying networks,'' \emph{{IEEE} Commun. Lett.}, vol.~19, no.~10, pp.
  1686--1689, Oct. 2015.


\bibitem{YNiu_Exploiting}
Y.~Niu, C.~Gao, Y.~Li, L.~Su, and D.~Jin, ``Exploiting multi-hop relaying to
  overcome blockage in directional mmwave small cells,'' \emph{Journal of
  Communications and Networks}, vol.~18, no.~3, pp. 364--374, June 2016.

\bibitem{JBKim_Nonorthogonal}
J.~B. Kim and I.~H. Lee, ``Non-orthogonal multiple access in coordinated direct
  and relay transmission,'' \emph{{IEEE} Commun. Lett.}, vol.~19, no.~11, pp.
  2037--2040, Nov. 2015.

\bibitem{CZong_Nonorthogonal_FD}
C.~Zhong and Z.~Zhang, ``Non-Orthogonal Multiple Access With Cooperative Full-Duplex Relaying,'' \emph{{IEEE} Commun. Lett.}, vol.~20, no.~12, pp.
  2478--2481, Dec. 2016.

\bibitem{CSPatel_Statistical}
C.~S. Patel, G.~L. Stuber, and T.~G. Pratt, ``Statistical properties of amplify
  and forward relay fading channels,'' \emph{{IEEE} Trans. Veh. Technol.},
  vol.~55, no.~1, pp. 1--9, Jan. 2006.

\bibitem{Table_Integral}
I.~S. Gradshteyn and I.~M. Ryzhik, \emph{Table of Integrals, Series, and
  Products}, 7th~ed.\hskip 1em plus 0.5em minus 0.4em\relax San Diego, CA:
  Academic, 2007.

\bibitem{Handbook_functions}
M.~Abramowitz and I.~A. Stegun, \emph{Handbook of mathematical functions: with
  formulas, graphs, and mathematical tables}, 1st~ed.\hskip 1em plus 0.5em
  minus 0.4em\relax Courier Corporation, 1964.

\end{thebibliography}




\end{document}